\documentclass[letterpaper, 10 pt, conference]{ieeeconf}  

\IEEEoverridecommandlockouts                              

\usepackage{amsmath, amssymb}

\usepackage{amsthm}

\usepackage{graphics} 
\usepackage{epsfig} 
\usepackage{mathptmx} 
\usepackage{times} 
\usepackage{algorithm}
\usepackage{algpseudocode}
\usepackage{booktabs}

\theoremstyle{plain}
\newtheorem{theorem}{Theorem}

\theoremstyle{definition}
\newtheorem{assumption}{Assumption}

\theoremstyle{remark}
\newtheorem{remark}{Remark}

\makeatletter
\renewenvironment{proof}[1][\proofname]{%
  \par\noindent\textit{#1. }\ignorespaces
}{%
  \hfill\qed\par
}
\makeatother

\title{\LARGE \bf
Safe Receding Horizon Mixed-Integer Differentiable Predictive Control for Degradation-Aware Battery Dispatch 
}

\author{
Eshagh Safarzadeh Ravajiri$^{1}$, 
J\'an Drgo\v na$^{2}$, 
and Benjamin F. Hobbs$^{1}$%
\thanks{%
$^{1}$Department of Environmental Health and Engineering, Johns Hopkins University. 
$^{2}$Department of Civil and Systems Engineering, Johns Hopkins University.%
}
}

\begin{document}

\maketitle
\thispagestyle{empty}
\pagestyle{empty}

\begin{abstract}
We present a safe receding-horizon mixed-integer differentiable
predictive control methodology for residential battery energy
storage dispatch that combines neural-network speed with recursive
feasibility guarantees.
Unlike open-loop learning-to-optimize methods, it incorporates
real-time state-of-charge feedback and sinusoidal time-of-day
conditioning, enabling closed-loop re-planning at every timestep
without re-solving a mixed-integer program.
A differentiable rainflow cycle-counting layer enables
self-supervised training of the mixed-integer policy on exact
degradation physics.
The controller is a hybrid closed-loop system pairing a neural
mode-selection and continuous-action policy with a
quadratic-programming safety filter that guarantees recursive
feasibility independent of network weights or mode optimality.
We establish mode-conditioned Lipschitz continuity and a conditional
regret decomposition into training-quality, mode-mismatch, and
forecast-error terms.
On a 7-day net-metering evaluation, the method attains a $6.9\%$ cost
gap versus the closed-loop mixed-integer MPC benchmark with a
$25\times$ speedup ($0.11$\,s vs $2.7$\,s per step), while average
regret rises by under $4\%$ across $0$--$30\%$ forecast noise. The
learned and benchmark modes agree at every step, so the bound reduces
to its training-quality and forecast-error terms, both empirically
validated.
\end{abstract}

\section{INTRODUCTION}
Battery energy storage systems (BESS) are central to residential and
utility-scale grid decarbonization, enabling arbitrage between
time-of-use price peaks and solar generation valleys.
Optimal dispatch requires solving a mixed-integer nonlinear program
(MINLP) at every control step. Charge/discharge mode selection is
discrete, and accurate degradation accounting requires the rainflow
cycle-counting algorithm~\cite{xu2018_liion_degradation_model}, which
is non-smooth and path-dependent.
Classical approaches linearize degradation into piecewise-linear
(PWL) segments and solve a MILP~\cite{xu2018_cycle_aging_market},
which is tractable but still re-solves the full MILP at every
timestep in closed loop.

To overcome this bottleneck, a growing body of work has turned to
learning-based approximate Model Predictive Control (MPC), employing
neural networks to imitate an expert optimal controller for large
inference-time speedups~\cite{chen2022learning, karg2020efficient,
hewing2020learning}, with recent extensions to mixed-integer settings
that map states directly to discrete or hybrid
actions~\cite{Cauligi2022, bertsimas2022online_mio, 11651352}.
However, these rely on supervised behavioral cloning and massive
offline datasets of optimal solutions. Differentiable Predictive
Control (DPC)~\cite{drgovna2024learning} instead departs from
imitation, differentiating the expected control objective through the
unrolled closed-loop dynamics to train policies end-to-end,
minimizing closed-loop cost and enforcing constraints without an
expert dataset.

Boldock\'{y} et al.~\cite{boldocky2025learning} extended DPC to
mixed-integer settings via differentiable rounding over integer and
continuous variables. Our prior work~\cite{Eshagh2026midpc} applied
mixed-integer DPC to BESS dispatch with a differentiable rainflow
layer, but modeled it as open-loop control, mapping a day-ahead
forecast to a fixed dispatch sequence with no state-of-charge (SOC)
feedback under process noise and forecast error.

This paper proposes a safe receding-horizon mixed-integer
differentiable predictive control (S-RH-MI-DPC) methodology for BESS
dispatch, enabling offline self-supervised training of a neural
policy via a differentiable rainflow layer, with feasible online
deployment through a convex QP safety filter under the
receding-horizon principle. Our contributions are:
\begin{enumerate}
    \item A real-time control framework for BESS dispatch that
    incorporates SOC feedback and sinusoidal time-of-day
    conditioning, extending our previous
    work~\cite{Eshagh2026midpc}.

    \item Formal proofs establishing the pointwise and recursive
    feasibility of the proposed S-RH-MI-DPC policy.

    \item A proof that the mode-conditioned policy, the neural
    network continuous-control action followed by the convex QP
    safety filter, maintains mode-conditioned Lipschitz continuity.

    \item A dynamic regret bound
    $|R_T| \leq L_J\bigl(\sqrt{T}\,\epsilon
      + \lVert\boldsymbol{\delta}^{\mathrm{mode}}\rVert_2\bigr)
      + b\sum_k e_k$, decomposing suboptimality into a
    training-quality term $\epsilon$, a discrete mode-mismatch term
    $\boldsymbol{\delta}^{\mathrm{mode}}$, and a cumulative
    forecast-error term $e_k$.
\end{enumerate}

\section{Problem Formulation}
\label{sec:problem}

\subsection{Dispatch Optimization Problem}

We consider a grid-connected BESS co-located with a rooftop solar
array under net metering, operating over a $T$-step horizon
($\Delta t =1$\,h, $T = 24$ per day).
Let $x_k \in \mathbb{R}$ denote the SOC at step $k$.
The optimization problem is
\begin{equation}
\begin{aligned}
  \min_{\mathbf{u}, \boldsymbol{\sigma}} \quad J = \sum_{k=0}^{T-1}
  &\bigl( p_k^{\mathrm{imp}} \pi_k^b
        - p_k^{\mathrm{exp}} \pi_k^s \bigr)\Delta t 
  + C_{\mathrm{cyc}}(x_{0:T})
\end{aligned}
\label{eq:opt_problem}
\end{equation}

\vspace{-1em}

\begin{subequations}
\label{eq:constraints}
\begin{align}
\text{subject to:}\quad
  & \sigma_k \in \Sigma :=
    \{\mathrm{charge},\,\mathrm{discharge},\,\mathrm{idle}\},
    && \forall k \label{eq:mode_ex} \\
  & u_{c,k} \leq P_{\max}\,\mathbf{1}[\sigma_k = \mathrm{charge}],
    && \forall k \label{eq:power_couple_c} \\
  & u_{d,k} \leq P_{\max}\,\mathbf{1}[\sigma_k = \mathrm{discharge}],
    && \forall k \label{eq:power_couple_d} \\
  & x_{k+1} = f_{\sigma_k}(x_k, u_k),
    && \forall k \label{eq:soc_dyn} \\
  & x_k \in \mathcal{X} := [x_{\min},\, x_{\max}],
    && \forall k \label{eq:soc_lim} \\
  & p_k^{\mathrm{imp}} - p_k^{\mathrm{exp}} =
    d_k - s_k + u_{c,k} - u_{d,k},
    && \forall k \label{eq:power_bal}
\end{align}
\end{subequations}

where $\pi_k^b, \pi_k^s$ are the time-of-use purchase and export
tariffs, $d_k, s_k$ the exogenous site demand and solar generation,
$\eta_c, \eta_d \in (0,1]$ the charge and discharge efficiencies, and
$P_{\max}, E_{\mathrm{cap}}$ the power rating and energy capacity.
The indicator $\mathbf{1}[\cdot]$ couples each mode to its admissible
power in \eqref{eq:power_couple_c}--\eqref{eq:power_couple_d}, and the
control vector is
$u_k := [u_{c,k},\, u_{d,k}]^\top \in \mathbb{R}^2$ with
$u_{c,k}, u_{d,k} \in [0, P_{\max}]$.
The net-metering balance \eqref{eq:power_bal} links the grid powers
$p_k^{\mathrm{imp}}, p_k^{\mathrm{exp}} \geq 0$ to net site load,
while $C_{\mathrm{cyc}}(x_{0:T})$ prices cycle aging via the rainflow
model of Section~\ref{sec:rainflow}. The mode-dependent dynamics are
\begin{equation}
  f_{\sigma}(x_k, u_k) = x_k +
  \begin{cases}
    \eta_c\, u_{c,k}\,\Delta t
      & \sigma = \mathrm{charge}, \\[3pt]
    -u_{d,k}/\eta_d\,\Delta t
      & \sigma = \mathrm{discharge}, \\[3pt]
    0 & \sigma = \mathrm{idle}.
  \end{cases}
  \label{eq:mode_dynamics}
\end{equation}

\subsection{Battery Degradation Model}
\label{sec:rainflow}
The degradation cost $C_{\mathrm{cyc}}$ is computed via the rainflow
cycle-counting algorithm~\cite{endo1968_rainflow} applied to
$x_{0:T}$, which extracts $N_{\mathrm{cyc}}$ half- and full-cycles,
each of depth $\delta_i$ and multiplicity $n_i \in \{0.5, 1\}$.
Each cycle is penalized by the stress
function~\cite{xu2018_cycle_aging_market}:
\begin{equation}
  \Phi(\delta) = a\,\delta^b,\qquad
  a = 5.24\times10^{-4},\quad b = 2.03,
  \label{eq:stress}
\end{equation}
for Li(NiMnCo)O$_2$ chemistry, giving
$C_{\mathrm{cyc}}(x_{0:T}) = \frac{R}{\eta_d}
  \sum_{i=1}^{N_{\mathrm{cyc}}} n_i\,\Phi(\delta_i)$,
where $R$ is the battery replacement cost.
Problem~\eqref{eq:opt_problem} is an MINLP due to the discrete modes
$\sigma_k$ and the nonlinear, nonsmooth, path-dependent cost
$C_{\mathrm{cyc}}$. The tractable
approximation~\cite{xu2018_cycle_aging_market} replaces
$C_{\mathrm{cyc}}$ with an $M$-segment piecewise-linear surrogate,
assigning segment $j$ the constant marginal cost
\begin{equation}
  C_j = \frac{R \cdot M}{\eta_d\,E_{\mathrm{cap}}}
        \Bigl[\Phi\!\Bigl(\tfrac{j}{M}\Bigr)
             -\Phi\!\Bigl(\tfrac{j-1}{M}\Bigr)\Bigr],
  \qquad j = 1,\dots,M,
  \label{eq:pwl_cost}
\end{equation}
converting~\eqref{eq:opt_problem} to a MILP.

 
\section{Methodology}
\label{sec:method}

The proposed S-RH-MI-DPC operates in two phases. Offline, the policy
$(\mu_\theta, \pi_\theta)$ is trained end-to-end by differentiating
the closed-loop objective through modes, power, SOC, and the rainflow
layer (Algorithm~\ref{alg:offline}). Online, it runs as a
receding-horizon controller (Algorithm~\ref{alg:online}): at each
hour it maps the true SOC $x_k$ and forecast window $\bar{\xi}_k$ to a
mode $\sigma_k$ and candidate $\hat{u}_k$ in one forward pass,
projects through a convex QP safety filter, and applies only the
first action. This decouples the cost of learning from the real-time
loop.

\begin{algorithm}[t]
\caption{S-RH-MI-DPC: Offline Training}
\label{alg:offline}
\begin{algorithmic}[1]
\Require Training data $\mathcal{D}$, epochs $E$, batch size $B$,
         degradation schedule $w_e$, temperature schedule $\tau$
\Ensure  Trained policy parameters $\theta^*$
\State Initialise $(\mu_\theta, \pi_\theta)$ with mode head
\For{epoch $= 1, \ldots, E$}
  \State Update $w_e$, $\tau$ via curriculum schedule
  \For{each mini-batch $\{s, d, \pi^b, \pi^s, x_0\}$ of size $B$}
    \State Sample $k_{\mathrm{start}} \sim \mathrm{Uniform}[0, T)$
    \State Construct $\phi_k$ from $(x_0, \bar{\xi})$
           via~\eqref{eq:input}
    \State Forward:
           $\hat{\boldsymbol{\sigma}} \leftarrow
           \mu_\theta(\phi_k)$ via Gumbel-Softmax,
           $\hat{\mathbf{u}} \leftarrow \pi_\theta(\phi_k)$
           via tanh scaling
    \State Unroll SOC: $x_{0:T} \leftarrow
           f_{\hat{\sigma}_k}(x_k, \hat{u}_k)$
           for $k = 0, \ldots, T-1$
    \State Compute loss:
           $\mathcal{L} = \mathcal{L}_{\mathrm{econ}} +
            w_e\,\mathcal{L}_{\mathrm{cyc}} +
            \lambda_{\mathrm{soc}}\,\mathcal{L}_{\mathrm{soc}}$
    \State Backward through differentiable rainflow layer
    \State Update $\theta$ via AdamW with gradient clipping
  \EndFor
\EndFor
\State \Return $\theta^* \leftarrow \theta$
\end{algorithmic}
\end{algorithm}

\begin{algorithm}[t]
\caption{S-RH-MI-DPC: Online Deployment}
\label{alg:online}
\begin{algorithmic}[1]
\Require Trained policy $(\mu_{\theta^*}, \pi_{\theta^*})$,
         initial SOC $x_0$, forecast oracle $\bar{\xi}_k$
\For{each timestep $k = 0, 1, 2, \ldots$}
  \State Observe true SOC $x_k$
  \State Query forecast window
         $\bar{\xi}_k \leftarrow \xi_{k:k+T-1}$
  \State Construct $\phi_k$ from $(x_k, \bar{\xi}_k)$
         via~\eqref{eq:input}
  \State \textbf{Mode selection:}
         $\sigma_k \leftarrow \mu_{\theta^*}(\phi_k)$
         via argmax
         \hfill $\triangleright$ $\sigma_k \in \Sigma$
  \State \textbf{Control candidate:}
         $\hat{u}_k \leftarrow \pi_{\theta^*}(\phi_k)$
         \hfill $\triangleright$ single forward pass, no online MILP
  \State \textbf{QP safety filter:}
         $u_k \leftarrow
         \Pi_{\mathcal{U}_{\sigma_k}(x_k)}(\hat{u}_k)$
         \hfill $\triangleright$ identity when $\hat{u}_k$ feasible
  \State \textbf{Apply first action:} execute $u_k$
  \State Advance state:
         $x_{k+1} \leftarrow f_{\sigma_k}(x_k, u_k)$
\EndFor
\end{algorithmic}
\end{algorithm}

\subsection{Architecture Overview}

S-RH-MI-DPC replaces the per-step MINLP with a neural network.
The forecast window $\bar{\xi}_k := [s_{k:k+T},\, d_{k:k+T},\,
\pi^b_{k:k+T},\, \pi^s_{k:k+T}]$ contains solar, demand, and
price profiles over the horizon.
The policy input $\phi_k$ concatenates the normalised forecast,
current SOC, and a sinusoidal time encoding:
\begin{equation}
\phi_k =
\begin{aligned}
\bigl[\, &
s_{k:k+T}/s_{\max},\;
d_{k:k+T}/d_{\max}, \\
& \pi^b_{k:k+T}/\pi^b_{\max},\;
\Delta\pi_{k:k+T}/\Delta\pi_{\max}, \\
& x_k/E_{\mathrm{cap}},\;
\sin\!\left(\tfrac{2\pi k}{T}\right),\;
\cos\!\left(\tfrac{2\pi k}{T}\right)
\,\bigr]
\end{aligned}
\label{eq:input}
\end{equation}
so that $(\sigma_k, \hat{u}_k) = (\mu_\theta(\phi_k),\,
\pi_\theta(\phi_k))$, with $\phi_k$ encoding both state and
forecast.
The SOC scalar $x_k/E_{\mathrm{cap}}$ provides closed-loop state
feedback, and the sinusoidal encoding $(\sin, \cos)$ marks where
the window starts in the day, enabling consistent behavior whether
the horizon begins at midnight, noon, or any intraday hour.
The input dimension is $4T + 3$.

A three-layer MLP encoder with LayerNorm and dropout ($p=0.1$)
maps $\phi_k$ to a hidden representation
$\mathbf{h} \in \mathbb{R}^{d}$ ($d = 1024$), from which two
linear heads produce:
(i)~a mode head $\mu_\theta$ outputting logits over
$\Sigma = \{\mathrm{charge},\,\mathrm{discharge},\,
\mathrm{idle}\}$ for all $T$ steps, and
(ii)~a power head $\pi_\theta$ outputting scaled
control magnitudes $u_{c,k}, u_{d,k} \in [0, P_{\max}]$.

\subsection{Mode and Power Heads Outputs}

Mode exclusivity is enforced via Gumbel-Softmax with the
straight-through estimator~\cite{tang2025l2o_minlp}. The temperature
$\tau$ anneals linearly from $5.0$ to $0.1$ during training, and
$\sigma_k$ is a hard argmax over the mode logits at inference.
The mode $\sigma_k$ gates the power outputs:
\begin{equation}
\begin{aligned}
u_{c,k} &= \mathbf{1}[\sigma_k = \mathrm{charge}]
  \cdot \frac{\tanh(h^c_k)+1}{2}\, P_{\max}, \\
u_{d,k} &= \mathbf{1}[\sigma_k = \mathrm{discharge}]
  \cdot \frac{\tanh(h^d_k)+1}{2}\, P_{\max},
\end{aligned}
\label{eq:power_gate}
\end{equation}
so idle forces $u_{c,k} = u_{d,k} = 0$ by construction,
satisfying~\eqref{eq:mode_ex} architecturally.
The $\tanh$ maps power-head logits to $[0, P_{\max}]$ with stronger
gradients near zero than sigmoid, improving capacity utilization.

\subsection{Differentiable Rainflow Layer}

The component enabling end-to-end training on exact degradation
physics is the differentiable rainflow layer from our prior
work~\cite{Eshagh2026midpc}. In the forward pass, given the SOC
trajectory $\mathbf{x}_{0:T}$ from the deterministic
unroll~\eqref{eq:soc_dyn}, the rainflow algorithm identifies, for each cycle $i$, a peak
$\tau_p^i$, valley $\tau_v^i$, depth
$\delta_i = |x_{\tau_p^i} - x_{\tau_v^i}|$, and multiplicity
$n_i \in \{0.5, 1.0\}$.
\begin{equation}
  \frac{\partial C_{\mathrm{cyc}}}{\partial x_k}
  = \lambda_e\,\nabla_{\mathrm{exact}}(k)
  + \lambda_p\,\nabla_{\mathrm{proxy}}(k),
  \qquad \lambda_e = \lambda_p = 0.5.
  \label{eq:hybrid_grad}
\end{equation}
The exact component is sparse, nonzero only at identified extrema,
\begin{equation}
\nabla_{\mathrm{exact}}(x_k)
= \sum_{i:\, k \in \{\tau_p^i,\, \tau_v^i\}}
\frac{\partial}{\partial x_k}
\left(
\frac{R\, a\, \delta_i^{\,b}}{E_{\mathrm{cap}}}
\right),
  \label{eq:exact_grad}
\end{equation}
providing physics-accurate gradients at the SOC extrema that define
cycle depths.
The proxy component densely differentiates the surrogate
\begin{equation}
  \tilde{C} = \frac{Ra}{E_{\mathrm{cap}}^b}
              \sum_k \bigl(|\Delta x_k| + \epsilon\bigr)^b
  \label{eq:proxy}
\end{equation}
via autograd, providing a nonzero gradient at every timestep.
All gradients are clipped to $[-1, 1]$ for stability.

\subsection{Training Procedure}

Training samples are random $T$-step windows from synthetic 2-day
sequences with randomized start times
$k_{\mathrm{start}} \sim \mathrm{Uniform}[0, T)$ and initial SOC
$x_0 \sim \mathrm{Uniform}[x_{\min}, x_{\max}]$; Gaussian noise
(5\%) is added to all input channels for robustness.
The self-supervised training loss is
\begin{equation}
  \mathcal{L} = \mathcal{L}_{\mathrm{econ}}
              + w_e\,\mathcal{L}_{\mathrm{cyc}}
              + \lambda_{\mathrm{soc}}\,\mathcal{L}_{\mathrm{soc}},
  \label{eq:loss}
\end{equation}
where $\mathcal{L}_{\mathrm{econ}}$ is import cost minus export
revenue, $\mathcal{L}_{\mathrm{cyc}} = C_{\mathrm{cyc}}$ is the
exact rainflow degradation cost, and $\mathcal{L}_{\mathrm{soc}} =
\mathrm{mean}(\mathrm{ReLU}(x_{\min} - x_k)
+ \mathrm{ReLU}(x_k - x_{\max}))$ penalizes SOC violations.

\subsection{Convex QP Safety Filter}

At inference, S-RH-MI-DPC applies the learned policy and projects
its output onto the one-step feasible set before execution:
\begin{equation}
  u_k = \Pi_{\mathcal{U}_{\sigma_k}(x_k)}(\hat{u}_k), \qquad
  \hat{u}_k = \pi_\theta(x_k, \bar{\xi}_k),
  \label{eq:safety_filter}
\end{equation}
where the projection solves the convex quadratic program
\begin{align}
\min_{u_{c,k},\, u_{d,k}} \quad
  & (u_{c,k} - \hat{u}_{c,k})^2
  + (u_{d,k} - \hat{u}_{d,k})^2 \label{eq:qp_obj} \\
\text{s.t.} \quad
  & x_{\min} \leq f_{\sigma_k}(x_k, u_k) \leq x_{\max},
    \label{eq:qp_bounds} \\
  & 0 \leq u_{c,k} \leq P_{\max},\;
    0 \leq u_{d,k} \leq P_{\max}, \label{eq:qp_power} \\
  & u_{c,k} = 0 \;\text{if}\; \sigma_k \neq \mathrm{charge},\;
    u_{d,k} = 0 \;\text{if}\; \sigma_k \neq \mathrm{discharge}.
    \label{eq:mode_lock}
\end{align}
solved via OSQP.
The mode locks~\eqref{eq:mode_lock} enforce exclusivity by zeroing
the power channel not selected by $\sigma_k$, so the projection
adjusts only the magnitude within the selected mode, never its
direction, preserving~\eqref{eq:mode_ex} exactly.
The QP is invoked only when
$f_{\sigma_k}(x_k, \hat{u}_k) \notin \mathcal{X}$; feasible outputs
pass through unchanged.

\section{Theoretical Analysis}
\label{sec:theory}

\subsection{Hybrid Closed-Loop System Definition}

We model the proposed controller as a switched (hybrid) dynamical
system driven by a learned mode-selection policy and a
mode-conditioned safety filter.
At each time step $k$, the closed-loop evolution is given by
\begin{align}
  \sigma_k &= \mu_\theta(x_k, \bar{\xi}_k) \in \Sigma,
    \label{eq:mode_select} \\
  \hat{u}_k &= \pi_\theta(x_k, \bar{\xi}_k),
    \label{eq:nn_action} \\
  u_k &= \Pi_{\mathcal{U}_{\sigma_k}(x_k)}(\hat{u}_k),
    \label{eq:project} \\
  x_{k+1} &= f_{\sigma_k}(x_k, u_k),
    \label{eq:switched_dyn}
\end{align}
where $\Sigma$, $u_k$, $f_{\sigma_k}$, and $\mathcal{X}$
are as defined in Section~\ref{sec:problem};
$\mu_\theta \colon \mathcal{X} \times \Xi \to \Sigma$ is the
neural network mode-selection map
(argmax over the mode head logits);
$\pi_\theta \colon \mathcal{X} \times \Xi \to \mathbb{R}^2$
generates continuous control candidates; and
$\Pi_{\mathcal{U}_{\sigma}(x)}(\cdot)$ denotes the Euclidean
projection onto the mode-conditioned feasible set
$\mathcal{U}_{\sigma}(x)$, separating discrete mode selection,
continuous control generation, and safe projection.

For mode $\sigma \in \Sigma$ at state $x_k \in \mathcal{X}$,
the mode-conditioned feasible set is
\begin{equation}
\mathcal{U}_\sigma(x_k) = \bigl\{\, u_k \in \mathbb{R}^2
\;\big|\;
0 \le u_k \le P_{\max},\;
u_k^{\bar{\sigma}} = 0,\;
f_\sigma(x_k, u_k) \in \mathcal{X}
\,\bigr\},
\label{eq:mode_feasible_set}
\end{equation}
where $u_k^{\bar{\sigma}} = 0$ zeroes the power channel not
selected by mode $\sigma$ (e.g.\ $u_{d,k} = 0$ in
charge mode), and $f_\sigma$ is the mode-dependent
dynamics~\eqref{eq:mode_dynamics}.

At step $k$, let
$\bar{\xi}_k := \xi_{k:k+T-1} \in \mathbb{R}^{n_\xi T}$
denote the length-$T$ forecast window,
$\Xi$ the compact set of admissible windows, and
$e_k := \lVert \bar{\xi}_k -
\bar{\xi}_k^{\mathrm{true}} \rVert_2$ the per-step
forecast error.
The trajectory cost from~\eqref{eq:opt_problem} is
restated in terms of the control sequence as
\begin{equation}
\begin{aligned}
J(\mathbf{u}_{0:T-1};\, x_0) :=
&\sum_{k=0}^{T-1} \Bigl(
    \pi^b_k\, p^{\mathrm{imp}}_k
    - \pi^s_k\, p^{\mathrm{exp}}_k
\Bigr) \Delta t \\
&+ C_{\mathrm{cyc}}\Bigl(
    \mathbf{x}_{0:T}(\mathbf{u}_{0:T-1}, x_0)
\Bigr),
\end{aligned}
\label{eq:traj_cost}
\end{equation}
where
$p^{\mathrm{imp}}_k = \max(0,\, d_k - s_k + u_{c,k} - u_{d,k})$
and
$p^{\mathrm{exp}}_k = \max(0,\, -(d_k - s_k + u_{c,k} - u_{d,k}))$.

\subsection{Feasibility of the Hybrid Controller}

\begin{assumption}[Mode-Conditioned Feasibility]
\label{ass:mode_feasibility}
For all $x \in \mathcal{X}$ and all $\sigma \in \Sigma$,
the mode-conditioned feasible set
$\mathcal{U}_\sigma(x)$ is nonempty, closed, and convex.
\end{assumption}

\begin{remark}
Assumption~\ref{ass:mode_feasibility} holds by construction
for the BESS system: no minimum power constraint is imposed,
so the zero-magnitude action $u_k = 0$ is admissible in every
mode, giving $f_{\sigma}(x_k, \mathbf{0}) = x_k \in \mathcal{X}$.
Convexity and closedness follow from linearity of $f_\sigma$ in
$u_k$ and compactness of the power and SOC bounds.
\end{remark}

\begin{theorem}[Well-Posed Hybrid Feasibility]
\label{thm:pointwise}
Under Assumption~\ref{ass:mode_feasibility}, for any
$x_k \in \mathcal{X}$ and any mode
$\sigma_k \in \Sigma$ selected by~\eqref{eq:mode_select},
the mode-conditioned projection~\eqref{eq:project} admits
a unique solution.
\end{theorem}

\begin{proof}
$\mathcal{U}_{\sigma_k}(x_k)$ is nonempty, closed, and convex
(Assumption~\ref{ass:mode_feasibility}), and the objective
$\tfrac{1}{2}\lVert u - \hat{u}_k \rVert_2^2$ is strictly convex, so
the minimiser is unique~\cite{bauschke2017convex}.
\end{proof}

\begin{theorem}[Recursive State Feasibility / Positive Invariance]
\label{thm:recursive}
Under Assumption~\ref{ass:mode_feasibility}, if
$x_0 \in \mathcal{X}$, then $x_k \in \mathcal{X}$ for all
$k \geq 0$ under the hybrid closed-loop
system~\eqref{eq:mode_select}--\eqref{eq:switched_dyn}.
\end{theorem}

\begin{proof}
By induction. The base case $x_0 \in \mathcal{X}$ holds by
assumption. Assuming $x_k \in \mathcal{X}$, Theorem~\ref{thm:pointwise}
returns $u_k \in \mathcal{U}_{\sigma_k}(x_k)$, and every element of
that set satisfies $f_{\sigma_k}(x_k, u_k) \in \mathcal{X}$
by~\eqref{eq:mode_feasible_set}, so
$x_{k+1} \in \mathcal{X}$~\cite{mayne2000constrained}.
\end{proof}

\begin{remark}
Theorem~\ref{thm:recursive} establishes one-step state
feasibility, and hence positive invariance of $\mathcal{X}$,
for \emph{any} switching signal $\{\sigma_k\}_{k \ge 0}$
produced by the neural network, independently of whether the
selected mode is globally optimal.
The result assumes exact nominal SOC dynamics; under an
additive process disturbance
$x_{k+1}=f_{\sigma_k}(x_k,u_k)+w_k$ it holds only with
robust tightening of $\mathcal{X}$ or a robust invariant set.
\end{remark}

\begin{theorem}[Mode-Conditioned Lipschitz Continuity]
\label{thm:lipschitz}
Suppose $\pi_\theta$ is $L_\pi$-Lipschitz on
$\mathcal{X} \times \Xi$.
For fixed $x_k \in \mathcal{X}$ and fixed mode
$\sigma_k \in \Sigma$, the composite controller
$\kappa_{\sigma_k}(x_k, \cdot) :=
\Pi_{\mathcal{U}_{\sigma_k}(x_k)} \circ
\pi_\theta(x_k, \cdot)$
is $L_\pi$-Lipschitz in $\bar{\xi}$:
\begin{equation}
  \lVert \kappa_{\sigma_k}(x_k, \bar{\xi})
       - \kappa_{\sigma_k}(x_k, \bar{\xi}') \rVert_2
  \leq L_\pi \lVert \bar{\xi} - \bar{\xi}' \rVert_2,
  \quad \forall\, \bar{\xi}, \bar{\xi}' \in \Xi.
  \label{eq:lipschitz}
\end{equation}
\end{theorem}

\begin{proof}
For fixed $x_k$ and $\sigma_k$, $\mathcal{U}_{\sigma_k}(x_k)$ is a
fixed nonempty closed convex set, so its projection is
non-expansive~\cite{bauschke2017convex}. Composing with the
$L_\pi$-Lipschitz $\pi_\theta(x_k,\cdot)$ gives
$L_\kappa \le L_\pi$.
\end{proof}

\begin{remark}
The bound~\eqref{eq:lipschitz} holds within a fixed mode;
the overall mixed-integer policy is therefore only
piecewise Lipschitz.
At mode transitions (where $\sigma_k$ changes) the policy may
exhibit discontinuities inherent to switched systems, so the
Lipschitz guarantee applies to forecast perturbations that do
not induce a mode switch.
\end{remark}


\subsection{Conditional Performance Bounds}
This subsection characterizes how suboptimality decomposes into
identifiable terms under stated assumptions, linking closed-loop
cost to training quality, discrete mode mismatch, and forecast
accuracy.

\begin{assumption}[Cost Regularity]
\label{ass:cost}
The trajectory cost $J$ defined in~\eqref{eq:traj_cost}
is $L_J$-Lipschitz in $\mathbf{u}_{0:T-1}$ on the feasible
control set, with a finite constant $L_J = L_J(T) < \infty$
that may depend on the horizon.
\end{assumption}

\begin{remark}
The constant $L_J(T)$ may depend on the horizon, the system
dynamics, the charge/discharge efficiencies, the energy prices, and
the regularity of the forward degradation functional; we do not
require a closed-form expression for it in the subsequent bounds.
As one contributing component, the economic stage cost applies
$\max(0,\cdot)$ to a linear function of
$u_k$, and, since the stage contains both import and export prices,
has per-stage constant
$L_{\ell,k} \le \sqrt{2}\,\max(\pi_k^b, \pi_k^s)\,\Delta t$ and
stacked economic constant
$L_{e,\mathrm{traj}} \le \sqrt{2T}\,\Delta t\,
\max_k \max(\pi_k^b,\pi_k^s)$.
The degradation functional $C_{\mathrm{cyc}}$ uses exact
rainflow counting in the forward pass with a hybrid surrogate
gradient used only for training; its training-time gradient
clipping does not certify regularity of the forward map, and exact
rainflow extraction is piecewise-defined, so we make no separate
Lipschitz claim for $C_{\mathrm{cyc}}$ and absorb its contribution
into $L_J(T)$ via Assumption~\ref{ass:cost}.
\end{remark}

\begin{assumption}[Training Approximation Quality]
\label{ass:eps}
Let $x_k^{\mathrm{NN}}$ denote the NN closed-loop state, and let
\begin{equation}
u_{k,\mathrm{err}}^*
:= \pi_{\mathrm{MPC}}\bigl(x_k^{\mathrm{NN}}, \bar{\xi}_k\bigr)
\end{equation}
be the MPC-PWL action under the (noisy) forecast $\bar{\xi}_k$
evaluated at the NN state.
There exists $\epsilon > 0$ such that
$\lVert \pi_\theta(x_k^{\mathrm{NN}}, \bar{\xi}_k)
  - u_{k,\mathrm{err}}^* \rVert_2 \leq \epsilon$ for all $k$.
\end{assumption}

\begin{remark}
Assumption~\ref{ass:eps} bounds the continuous imitation
error only; it does not assert that the learned mode matches the
comparator mode. When the two modes disagree,
$u_{k,\mathrm{err}}^*$ may lie outside
$\mathcal{U}_{\sigma_k}(x_k^{\mathrm{NN}})$, and the resulting
discrete error is captured by the mode-mismatch term
$\delta_k^{\mathrm{mode}}$ below.
\end{remark}

\begin{assumption}[Forecast-Induced Deviation]
\label{ass:delta}
Let $u_{k,\mathrm{true}}^* :=
\pi_{\mathrm{MPC}}(x_k^{\mathrm{NN}}, \bar{\xi}_k^{\mathrm{true}})$
be the MPC-PWL action under the true forecast, also evaluated at
the NN state, and write $\mathbf{u}^*_{\mathrm{err}}$,
$\mathbf{u}^*_{\mathrm{true}}$ for the corresponding sequences.
There exists $b > 0$ such that
\begin{equation}
  \bigl| J(\mathbf{u}^*_{\mathrm{err}};\, x_0) -
         J(\mathbf{u}^*_{\mathrm{true}};\, x_0) \bigr|
  \leq b \sum_{k=0}^{T-1} e_k.
\end{equation}
\end{assumption}

\begin{remark}
For the mixed-integer MPC-PWL comparator,
Assumption~\ref{ass:delta} should be read as a local
regularity condition where it holds over forecast neighborhoods on
which the optimal mode sequence is unchanged and the continuous
subproblem satisfies standard sensitivity
conditions~\cite{lin2022boundedregretmpcperturbationanalysis},
so that forecast errors perturb the solution and propagate to
bounded cost differences via Assumption~\ref{ass:cost}.
A perturbation that flips the optimal integer sequence can induce
a jump, so the bound is not claimed globally.
\end{remark}

\begin{theorem}[Trajectory Tracking Bound]
\label{thm:one_step}
Under Assumptions~\ref{ass:cost}--\ref{ass:eps}, define the
per-step mode-mismatch
$\delta_k^{\mathrm{mode}} := \operatorname{dist}\!\bigl(
u_{k,\mathrm{err}}^*,\,
\mathcal{U}_{\sigma_k}(x_k^{\mathrm{NN}})\bigr)$ and
$\lVert \boldsymbol{\delta}^{\mathrm{mode}} \rVert_2 :=
\bigl(\sum_{k=0}^{T-1} (\delta_k^{\mathrm{mode}})^2\bigr)^{1/2}$.
Then the per-step tracking error satisfies
$\lVert u_k - u_{k,\mathrm{err}}^* \rVert_2 \leq \epsilon +
\delta_k^{\mathrm{mode}}$, and the trajectory cost obeys
\begin{equation}
  \bigl| J(\mathbf{u}^{\mathrm{NN}}_{0:T-1};\, x_0) -
         J(\mathbf{u}^*_{\mathrm{err},0:T-1};\, x_0) \bigr|
  \leq L_J \bigl( \sqrt{T}\, \epsilon
        + \lVert \boldsymbol{\delta}^{\mathrm{mode}} \rVert_2 \bigr).
  \label{eq:one_step}
\end{equation}
\end{theorem}
\begin{proof}
Writing $C_k := \mathcal{U}_{\sigma_k}(x_k^{\mathrm{NN}})$ and
inserting $\Pi_{C_k}(u_{k,\mathrm{err}}^*)$, the triangle
inequality and non-expansiveness of the projection give
\begin{align}
  \lVert u_k - u_{k,\mathrm{err}}^* \rVert_2
  &= \bigl\lVert \Pi_{C_k}(\hat{u}_k)
        - u_{k,\mathrm{err}}^* \bigr\rVert_2
  \notag \\
  &\leq \bigl\lVert \Pi_{C_k}(\hat{u}_k)
        - \Pi_{C_k}(u_{k,\mathrm{err}}^*) \bigr\rVert_2
  \notag \\
  &\qquad + \bigl\lVert \Pi_{C_k}(u_{k,\mathrm{err}}^*)
        - u_{k,\mathrm{err}}^* \bigr\rVert_2
  \notag \\
  &\leq \lVert \hat{u}_k - u_{k,\mathrm{err}}^* \rVert_2
      + \operatorname{dist}(u_{k,\mathrm{err}}^*, C_k)
  \notag \\
  &\leq \epsilon + \delta_k^{\mathrm{mode}} .
\end{align}
The cost bound follows from Assumption~\ref{ass:cost}.
\end{proof}

\begin{theorem}[Dynamic Regret Bound]
\label{thm:regret}
Define the tracking regret
$R_T := J(\mathbf{u}^{\mathrm{NN}}_{0:T-1};\, x_0)
       - J(\mathbf{u}^*_{\mathrm{true},0:T-1};\, x_0)$.
Under Assumptions~\ref{ass:cost}--\ref{ass:delta},
\begin{equation}
  |R_T| \leq L_J \bigl( \sqrt{T}\, \epsilon
            + \lVert \boldsymbol{\delta}^{\mathrm{mode}} \rVert_2
            \bigr)
            + b \sum_{k=0}^{T-1} e_k.
  \label{eq:regret_bound}
\end{equation}
\end{theorem}
\begin{proof}
Insert the intermediate term
$J(\mathbf{u}^*_{\mathrm{err}};\, x_0)$ and apply the triangle
inequality:
\begin{align}
  |R_T|
  &\leq
  \underbrace{\bigl| J(\mathbf{u}^{\mathrm{NN}};\, x_0) -
    J(\mathbf{u}^*_{\mathrm{err}};\, x_0) \bigr|}_{
    \leq\; L_J(\sqrt{T}\,\epsilon
      + \lVert \boldsymbol{\delta}^{\mathrm{mode}} \rVert_2)
    \;\text{(Thm.~\ref{thm:one_step})}}
  +
  \underbrace{\bigl| J(\mathbf{u}^*_{\mathrm{err}};\, x_0) -
    J(\mathbf{u}^*_{\mathrm{true}};\, x_0) \bigr|}_{
    \leq\; b\sum_k e_k
    \;\text{(Asm.~\ref{ass:delta})}}.
\end{align}
\end{proof}

\begin{remark}[Mode-Agreement Special Case]
When the learned mode matches the comparator mode along the
closed-loop trajectory, $\sigma_k = \sigma_k^*$ for all $k$, then
$u_{k,\mathrm{err}}^* \in \mathcal{U}_{\sigma_k}(x_k^{\mathrm{NN}})$
and $\boldsymbol{\delta}^{\mathrm{mode}} = \mathbf{0}$, recovering
the clean bound
$|R_T| \leq L_J \sqrt{T}\,\epsilon + b \sum_k e_k$.
Because $\pi_\theta$ and the MPC-PWL comparator optimize the same
economic--degradation objective, their discrete decisions coincide at
most steps once training has converged, so
$\lVert \boldsymbol{\delta}^{\mathrm{mode}} \rVert_2$ is typically
small.
\end{remark}

\begin{remark}[Average Regret]
Dividing~\eqref{eq:regret_bound} by $T$,
\begin{equation}
  \frac{|R_T|}{T} \leq
    \frac{L_J\bigl(\sqrt{T}\,\epsilon
      + \lVert \boldsymbol{\delta}^{\mathrm{mode}} \rVert_2\bigr)}{T}
    + \frac{b}{T}\sum_{k=0}^{T-1} e_k.
  \label{eq:avg_regret}
\end{equation}
The first term vanishes as training converges ($\epsilon \to 0$)
and as mode agreement improves
($\lVert \boldsymbol{\delta}^{\mathrm{mode}} \rVert_2 \to 0$); its
horizon dependence is governed by $L_J = L_J(T)$, so we make no
unconditional $1/\sqrt{T}$ improvement claim.
The second term is the time-averaged forecast error, vanishing as
$e_k \to 0$.
Because MPC-PWL uses a piecewise-linear degradation approximation, the bound measures suboptimality relative to this baseline, not the global MINLP optimum; all terms are empirically validated in Section~\ref{sec:results}.
\end{remark}

\section{Experimental Results}
\label{sec:results}

We evaluate S-RH-MI-DPC on a BESS with a capacity
$E_{\mathrm{cap}} = 10$\,kWh, maximum charge/discharge power
$P_{\max} = 4.8$\,kW, efficiencies $\eta_c = \eta_d = 0.92$,
and SOC limits $[0.1, 0.9]\,E_{\mathrm{cap}}$.
The policy is trained with AdamW and OneCycleLR
($\mathrm{lr}_{\max} = 5\times10^{-4}$), gradient clipping
($\|\nabla\|_2 \leq 1.0$), degradation weight $w_e$
increasing from $0.1$ to $1.0$ via curriculum schedule,
SOC penalty weight $\lambda_{\mathrm{soc}} = 5$, for
1{,}500 epochs on 10{,}000 samples with batch size 48.

\subsection{Weekly Cost Performance}

Table~\ref{tab:weekly} summarises the 7-day cost breakdown.
S-RH-MI-DPC reaches a total weekly cost of \$41.33 versus \$38.65 for
MPC-PWL, a $6.9\%$ gap (\$2.68), with balanced economic and
degradation components confirming that the differentiable rainflow
layer trains on exact degradation physics.
All 168 steps are feasible ($100\%$), guaranteed by the QP safety
filter (invoked on $39.3\%$ of steps) regardless of NN output
quality.
MPC-PWL requires 451\,s (2.7\,s/step) as a branch-and-bound MILP,
whereas S-RH-MI-DPC requires 18.1\,s (0.11\,s/step), a $25\times$
speedup with fixed per-step latency and no online branch-and-bound. 
The learned mode matched the MPC-PWL comparator mode at all
steps ($100\%$ mode agreement), so the empirical mode-mismatch term
$\lVert \boldsymbol{\delta}^{\mathrm{mode}} \rVert_2 = 0$.

\begin{table}[t]
\centering
\caption{Weekly cost and timing (7 days).}
\label{tab:weekly}
\begin{tabular}{lccccc}
\toprule
Method & Econ. & Deg. & Total & vs $J^*$ & Time/step \\
\midrule
MPC-PWL   & \$32.10 & \$6.55 & \$38.65 & $0.0\%$  & 2.70\,s \\
S-RH-MI-DPC & \$35.24 & \$6.10 & \$41.33 & $+6.9\%$ & 0.11\,s \\
\bottomrule
\end{tabular}
\end{table}

\subsection{Experiment A: Regret vs Training Quality}

To validate the training-quality dependence of the regret
bound~\eqref{eq:avg_regret}, we evaluate S-RH-MI-DPC checkpoints
across training against $J^* = \$38.65$ (MPC-PWL).
Table~\ref{tab:exp_a} shows $R_T/T$ decreasing monotonically with
epochs, a $36\%$ reduction, as the action deviation
$\epsilon_{\mathrm{emp}}$ falls.
The near-linear $\epsilon_{\mathrm{emp}}$--$R_T/T$ relationship in
Fig.~\ref{fig:noise_soc}(a) is consistent with the training-quality
term of Theorem~\ref{thm:one_step} contracting as $\epsilon \to 0$.

\begin{table}[t]
\centering
\caption{Experiment A: average regret vs training epochs;
         $\epsilon_{\mathrm{emp}}$ is the mean per-step deviation
         from MPC-PWL.}
\label{tab:exp_a}
\begin{tabular}{rrrr}
\toprule
Epochs & $R_T$ (\$) & $R_T/T$ (\$/step) & $\epsilon_{\mathrm{emp}}$ \\
\midrule
 100  & $+4.186$ & $+0.0249$ & $0.816$ \\
 300  & $+3.936$ & $+0.0234$ & $0.791$ \\
 500  & $+3.703$ & $+0.0220$ & $0.775$ \\
 750  & $+3.040$ & $+0.0181$ & $0.718$ \\
1000  & $+2.874$ & $+0.0171$ & $0.710$ \\
1500  & $+2.682$ & $+0.0160$ & $0.693$ \\
\bottomrule
\end{tabular}
\end{table}

\subsection{Experiment B: Forecast Robustness}

To validate the forecast-error term, we inject multiplicative
Gaussian noise ($\sigma \in [0, 0.30]$) into solar and demand
forecasts, averaging $R_T/T$ over 5 seeds.
Table~\ref{tab:exp_b} shows $R_T/T$ nearly flat across the range, a
$3.75\%$ rise in regret for a $30\%$ increase in forecast error.
The receding-horizon policy observes the true SOC $x_k$ and replans
each step, so zero-mean noise is partly absorbed by state feedback;
the sensitivity constant $b$ is thus empirically small and the
forecast term is dominated by the training-quality term
(Fig.~\ref{fig:noise_soc}(b)).

\begin{table}[t]
\centering
\caption{Experiment B: $R_T/T$ vs forecast noise;
         $\bar{e} = \|\xi_{\mathrm{noisy}} - \xi_{\mathrm{true}}\|_2/T$.}
\label{tab:exp_b}
\begin{tabular}{ccc}
\toprule
Noise $\sigma$ & $\bar{e}$ & $R_T/T$ (\$/step) \\
\midrule
0.00 & 0.000 & $+0.0160 \pm 0.0000$ \\
0.05 & 0.030 & $+0.0160 \pm 0.0001$ \\
0.10 & 0.059 & $+0.0161 \pm 0.0002$ \\
0.15 & 0.089 & $+0.0162 \pm 0.0002$ \\
0.20 & 0.119 & $+0.0163 \pm 0.0003$ \\
0.25 & 0.148 & $+0.0164 \pm 0.0003$ \\
0.30 & 0.178 & $+0.0166 \pm 0.0004$ \\
\bottomrule
\end{tabular}
\end{table}

\begin{figure}[t]
  \centering
  \includegraphics[width=0.8\columnwidth]{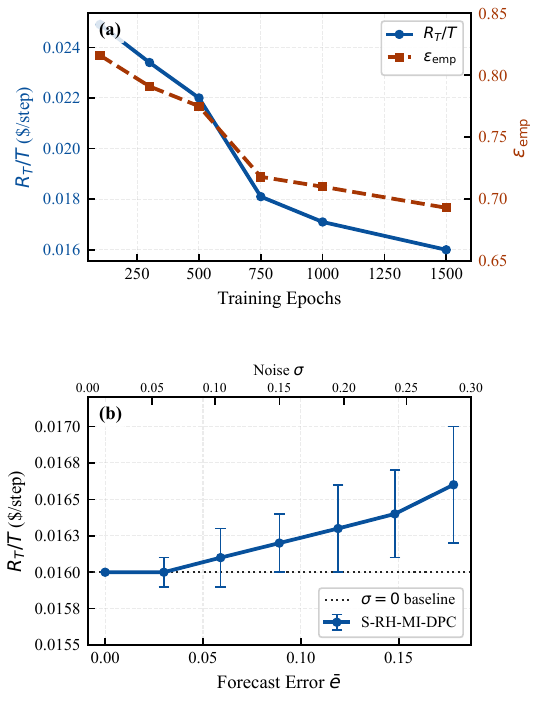}
  \caption{Empirical validation of~\eqref{eq:avg_regret}:
         (a)~regret vs.\ training epochs;
         (b)~regret vs.\ forecast noise.}
  \label{fig:noise_soc}
\end{figure}

\section{Conclusion}
\label{sec:conclusion}

We presented a safe receding-horizon mixed-integer differentiable
predictive control (S-RH-MI-DPC) methodology for battery energy
storage dispatch that replaces the per-step mixed-integer solve with
a pre-trained neural policy and a QP-based safety filter.
The controller is formally defined as a hybrid closed-loop system
with learned mode selection and mode-conditioned projection.
Theoretically, we establish recursive feasibility and positive
invariance of the SOC constraint set independent of network weights
or mode selection, and conditional performance bounds that decompose
suboptimality into training-quality, mode-mismatch, and
forecast-error terms.
These bounds are empirically validated where regret decreases
monotonically with epochs and increases by less than $4\%$ across
$0$--$30\%$ forecast noise.
On a 7-day evaluation, S-RH-MI-DPC achieves a $6.9\%$ cost gap versus
the mixed-integer MPC benchmark with $100\%$ feasibility and a
$25\times$ speedup ($0.11$\,s vs $2.7$\,s per step).
Extensions to multi-battery fleet coordination and real-time ancillary
service markets are planned.
\section*{Acknowledgements}
This research is partially supported by the U.S. DOE, Office of Science, ASCR program under the Scientific Discovery through Advanced Computing (SciDAC) Institute “LEADS: LEarning-Accelerated Domain Science”, and by the Ralph O’Connor Sustainable Energy Institute (ROSEI) at Johns Hopkins University.

\bibliographystyle{IEEEtran}
\bibliography{references}

@article{xu2018_liion_degradation_model,
  title   = {Modeling of Lithium-Ion Battery Degradation for Cell Life Assessment},
  author  = {Xu, Bolun and Oudalov, Alexandre and Ulbig, Andreas and Andersson, G{\"o}ran and Kirschen, Daniel S.},
  journal = {IEEE Transactions on Smart Grid},
  volume  = {9},
  number  = {2},
  pages   = {1131--1140},
  year    = {2018},
  doi     = {10.1109/TSG.2016.2578950},
  issn    = {1949-3053}
}

@ARTICLE{xu2018_cycle_aging_market,
  author={Xu, Bolun and Zhao, Jinye and Zheng, Tongxin and Litvinov, Eugene and Kirschen, Daniel S.},
  journal={IEEE Transactions on Power Systems}, 
  title={Factoring the Cycle Aging Cost of Batteries Participating in Electricity Markets}, 
  year={2018},
  volume={33},
  number={2},
  pages={2248-2259},
  doi={10.1109/TPWRS.2017.2733339}}

@article{chen2022learning,
  author  = {Tianlong Chen and Xiaohan Chen and Wuyang Chen and Howard Heaton and Jialin Liu and Zhangyang Wang and Wotao Yin},
  title   = {Learning to Optimize: A Primer and A Benchmark},
  journal = {Journal of Machine Learning Research},
  year    = {2022},
  volume  = {23},
  number  = {189},
  pages   = {1--59},
}

@article{drgovna2024learning,
  title={Learning constrained parametric differentiable predictive control policies with guarantees},
  author={Drgo{\v{n}}a, J{\'a}n and Tuor, Aaron and Vrabie, Draguna},
  journal={IEEE Transactions on Systems, Man, and Cybernetics: Systems},
  volume={54},
  number={6},
  pages={3596--3607},
  year={2024},
  publisher={IEEE}
}

@article{boldocky2025learning,
  title={Learning to solve parametric mixed-integer optimal control problems via differentiable predictive control},
  author={Boldock{\'y}, J{\'a}n and Javan, Shahriar Dadras and Gulan, Martin and M{\"o}nnigmann, Martin and Drgo{\v{n}}a, J{\'a}n},
  journal={arXiv:2506.19646},
  year={2025}
}

@article{bertsimas2022online_mio,
author = {Bertsimas, Dimitris and Stellato, Bartolomeo},
title = {Online Mixed-Integer Optimization in Milliseconds},
journal = {INFORMS Journal on Computing},
volume = {34},
number = {4},
pages = {2229-2248},
year = {2022},
doi = {10.1287/ijoc.2022.1181},
}

@ARTICLE{Cauligi2022,
  author={Cauligi, Abhishek and Culbertson, Preston and Schmerling, Edward and Schwager, Mac and Stellato, Bartolomeo and Pavone, Marco},
  journal={IEEE Robotics and Automation Letters}, 
  title={CoCo: Online Mixed-Integer Control Via Supervised Learning}, 
  year={2022},
  volume={7},
  number={2},
  pages={1447-1454},
  doi={10.1109/LRA.2021.3135931}}

@inproceedings{endo1968_rainflow,
  title     = {Damage Evaluation of Metals for Random or Varying Loading—Three Aspects of Rain Flow Method},
  author    = {Endo, Tatsuo and Mitsunaga, K. and Takahashi, K. and Kobayashi, K. and Matsuishi, Masanori},
  booktitle = {Proceedings of the 1974 Symposium on Mechanical Behaviour of Materials},
  volume    = {1},
  pages     = {371--380},
  year      = {1974}
}

@article{tang2025l2o_minlp,
  title   = {Learning to Optimize for Mixed-Integer Nonlinear Programming with Feasibility Guarantees},
  author  = {Tang, Bohan and Khalil, Elias B. and Drgo{\v{n}}a, J{\'a}n},
  journal = {arXiv preprint},
  volume  = {arXiv:2410.11061},
  year    = {2025},
  doi     = {10.48550/arXiv.2410.11061},
}

@book{bauschke2017convex,
  title        = {Convex Analysis and Monotone Operator Theory in Hilbert Spaces},
  author       = {Heinz H. Bauschke and Patrick L. Combettes},
  series       = {CMS Books in Mathematics},
  publisher    = {Springer Cham},
  year         = {2017},
  doi          = {10.1007/978-3-319-48311-5},
  isbn         = {978-3-319-48310-8},
  edition      = {2},
  pages        = {XIX, 619}
}

@article{mayne2000constrained,
  title        = {Constrained Model Predictive Control: Stability and Optimality},
  author       = {D.Q. Mayne and J.B. Rawlings and C.V. Rao and P.O.M. Scokaert},
  journal      = {Automatica},
  volume       = {36},
  number       = {6},
  pages        = {789--814},
  year         = {2000},
  doi          = {10.1016/S0005-1098(99)00214-9}
}

@article{karg2020efficient,
  title={Efficient Representation and Approximation of Model Predictive Control Laws via Deep Learning},
  author={Karg, Benjamin and Lucia, Sergio},
  journal={IEEE Transactions on Cybernetics},
  volume={50},
  number={9},
  pages={3866--3878},
  year={2020},
  doi={10.1109/TCYB.2020.2999556}
}

@inproceedings{lin2022boundedregretmpcperturbationanalysis,
 author = {Lin, Yiheng and Hu, Yang and Qu, Guannan and Li, Tongxin and Wierman, Adam},
 booktitle = {Advances in Neural Information Processing Systems},
 doi = {10.52202/068431-2621},
 pages = {36174--36187},
 title = {Bounded-Regret MPC via Perturbation Analysis: Prediction Error, Constraints, and Nonlinearity},
 volume = {35},
 year = {2022}
}

@article{hewing2020learning,
  title   = {Learning-Based Model Predictive Control: Toward Safe Learning in Control},
  author  = {Hewing, Lukas and Wabersich, Kim P. and Menner, Marcel and Zeilinger, Melanie N.},
  journal = {Annual Review of Control, Robotics, and Autonomous Systems},
  volume  = {3},
  pages   = {269--296},
  year    = {2020},
  doi     = {10.1146/annurev-control-090419-075625},
  issn    = {2573-5144}
}

@article{Eshagh2026midpc,
  title   = {End-to-End Battery Dispatch with Exact Rainflow Degradation via Mixed-Integer Differentiable Predictive Control},
  author  = {Eshagh Safarzadeh Ravajiri and Jan Drgona and Mahdi Mehrtash and Benjamin F. Hobbs},
  journal = {arXiv preprint},
  volume  = {arXiv:2609.12968},
  year    = {2026},
  doi     = {10.48550/arXiv.2609.12968},
}

@ARTICLE{11651352,
  author={Zheng, Honghui and Favaro, Pietro and Dvorkin, Yury and Drgona, Ján},
  journal={IEEE Transactions on Sustainable Energy}, 
  title={Accelerating Underground Pumped Hydro Energy Storage Scheduling with Decision-Focused Learning}, 
  year={2026},
  volume={},
  number={},
  pages={1-13},
  doi={10.1109/TSTE.2026.3722492}}

\end{document}